\documentclass[12pt]{report}

\usepackage{geometry}          
\usepackage[utf8]{inputenc}     
\usepackage{palatino}           
  
\usepackage{graphicx}           
\usepackage{indentfirst}       
\usepackage[nottoc]{tocbibind}  
\usepackage[unicode]{hyperref}  
\usepackage{amsmath,amsfonts,amsthm}
\usepackage{enumerate}

\newtheorem{theorem}{Theorem}[section]
\newtheorem{corollary}{Corollary}[theorem]
\newtheorem{lemma}[theorem]{Lemma}
\newtheorem{problem}{Problem}

\graphicspath{{images/}}    
\begin{document}
    
    \newcommand{\university}    {Chennai Mathematical Institute}
\newcommand{\dates}   {June 15, 2020}

\newcommand{\thesistype}    {M.Sc. Thesis}
\newcommand{\thesistitle}   {LTL with Local and Remote Data Constraints}

\newcommand{\supervisorlast}    {Praveen}                             
\newcommand{\supervisorfirst}   {M.}
\newcommand{\authorlast}    {Bhaskar}                              
\newcommand{\authorfirst}   {Ashwin}
\newcommand{\authornamefl}  {\authorfirst \space \authorlast} 
\newcommand{\supervisornamefl}  {Supervisor: \supervisorfirst \space \supervisorlast}
\newcommand{\authornamelf}  {\authorlast \space \authorfirst}

\newcommand{\dottedline}    {............................}

    \pagenumbering{gobble}

    \begin{titlepage}
    \begin{center}
       
        \vspace{15cm}
        \huge
        \textbf{\thesistitle}

        \vspace{2cm}
        \LARGE
        \textbf{\authornamefl}
        
        \vspace{0.05cm}
        \LARGE
        \text{\thesistype}
        
        \vspace{6cm}
        \LARGE
        \textbf{\supervisornamefl}

        \vfill
        \Large

        \vspace{0.5cm}
        \LARGE
        \text{\university}
        
        \vspace{0.5cm}
        \LARGE
        \text{\dates}
    \end{center}
\end{titlepage}

    \chapter*{Abstract} 
\begin{flushleft}
We consider an extension of linear-time temporal logic (LTL) with both local and remote data constraints interpreted over a concrete domain. This extension is a natural extension of constraint LTL and the Temporal Logic of Repeating Values, which have been studied before. We shall use previous results to prove that the satisfiability problem for this logic is decidable. Further, we shall see that trying to extend this logic by making it more expressive can lead to undecidability. 
\end{flushleft}
 
    \chapter*{Acknowledgements}

\begin{flushleft}
I would like to thank my supervisor, Prof. M. Praveen, for initiating me to the field of Logic, for helping me choose my thesis subject and for guiding me at each step, giving me the right inputs whenever needed. 
\end{flushleft}

\begin{flushleft}
I also wish to thank Prof. Srivathsan for motivating me to pursue the field of Logic and Automata Theory, and for helping me nurture my presentation and report writing skills.
\end{flushleft}

\begin{flushleft}
I also wish to thank Prof. Madhavan Mukund, Prof. Aiswarya Cyriac, Prof. S P Suresh and Prof. Narayan Kumar for the wonderful courses in Logic, Automata Theory and Verification that I took under them. 
\end{flushleft}

    \tableofcontents

    \setcounter{page}{1}
    \pagenumbering{arabic}

    \chapter{Introduction} 

Linear Temporal Logic (LTL) is a logic that is largely used to specify the properties/behaviour of a transition system. LTL is defined over boolean/propositional variables. Note that the values of these boolean variables range over the set $\{\top, \bot\}$. One way of extending LTL could be to define a logic over data variables whose values, in general could range over any infinite set. Now the models of such a logic would be valuation sequences of the variables over this infinite set. Again, there could be several such extensions possible. We are interested in two such extensions which we shall now breifly describe. \\

The propositional variables in the syntax of LTL usually stand for some properties. In case, the LTL formula is built to specify properties of a program, these propositions would capture the program states. In this case, these propositions could stand for properties like "the value of variable $x$ is non-negative", or that "the value of $x$ in the currrent state is less than the value of $y$ in the next state." In a natural extension of the linear temporal logic, one may consider allowing assertions directly on the value of variables, as in “$x \geq 0$” or “$x < \textnormal{\textbf{X}}y$”. The type of variables and the kind of constraints allowed, leads us to the study of constraint temporal logics. We are interested in constraint temporal logics which are parameterised by a constraint system $\mathcal{D}$ which comprises a concrete domain and an interpretation for relations, as introduced in \cite{DD07}.\\

The Constraint Linear Temporal Logic CLTL$(\mathcal{D})$ is essentially obtained from LTL by replacing propositions by atomic constraints in $\mathcal{D}$. In classical LTL variables represent propositions and models for its formulae are sequences of propositional valuations for these variables. These models can be viewed as having a “spatial” axis (here the elements true and false), along which the variables move. In constraint logics, the spatial axis for the models will comprise elements from the domain of $\mathcal{D}$. For example, with the constraint system $(\mathbb{N},<,\approx)$ one is allowed to use atomic constraints involving $<$ and $\approx$ (denotes the equality predicate), and variables which range over natural numbers. The formula $\textnormal{\textbf{G}}(x < y)$ in the logic parameterised by $(\mathbb{N},<,\approx)$ is interpreted over a sequence of $\mathbb{N}$-valuations for the variables $x$ and $y$, and asserts that at every point in the future, the value of the variable $x$ is less than the value of $y$. This formula is of course satisfiable, a candidate satisfying model being $ss...$, where the valuation $s$ assigns 1 to $x$ and 2 to $y$. Note that the constraints that are allowed in this logic are local constraints. This means that one cannot assert propositions like the value of $x$ at the current state is equal to the value of $y$ in some future state. \\

The other extension that we are interested in, is the Temporal Logic of Repeating Values, CLTL$^{\text{XF}}$ introduced in \cite{DDG12}. This logic allows us to reason about repetitions of data values from an infinite domain $D$. In this logic, the propositions are replaced by atomic constraints over an infinite set $D$. As before, the spatial axis for the models will comprise elements from the domain $D$. However, unlike in the case of CLTL$(\mathcal{D})$, the only predicate allowed in these constraints is the equality predicate $(\approx)$. But, as opposed to CLTL$(\mathcal{D})$, both local and remote constraints are allowed in this logic.  That is, one can assert that the value of $x$ is equal to the value of $y$ after 5 states. One can also assert that the value of $x$ is equal to the value of $y$ in some future state. \\

In this thesis, we shall consider a natural extension of the logics described above. That is, we would consider a logic that allows both local and remote constraints, and also allows the $<$ predicate. However, as we shall see, we must be careful while defining such a logic. If we construct a logic that allows us to assert propositions such as: "The value of $x$ is less than the value of $y$ at some future position", then the logic turns out to be undecidable, as we shall see in this thesis. However, we see that if we construct a careful natural extension of the two logics by defining it over a disjoint union of local and remote variables, in such a manner that: both local and remote, $\approx$ comparisons are allowed over the remote variables; only local, $<$ and $\approx$ comparisions are allowed over the local variables and; no interaction is allowed between the local and remote variables, then, the satisfiability problem for the logic so constructed turns out to be decidable. We shall even see that allowing for any interaction between the local and remote variables in this logic leads to undecidability. \\

\textnormal{\textbf{Thesis Outline}}: In Chapter 2, we present the preliminaries which we will
need throughout the thesis. We present the syntax and semantics of the logics CLTL$(\mathcal{D})$ (\cite{DD07}) and CLTL$^{\text{XF}}$ (\cite{DDG12}), recall some necessary definitions and present some of the important results related to the satisfiability problems for these logics. In Chapter 3, we introduce the logic CLRV$^{\textnormal{XF}}$($\mathcal{D}$) and using previous results, prove that the satisfiability problem for this logic is decidable over dense, open domains. In Chapter 4, we introduce the logics LRC$^{\top}(\mathcal{D})$ and CLRV$^{\textnormal{XF},\textnormal{XF}^{-1}, \textnormal{int}}$($\mathcal{D}$) and prove that the satisfiability problems for these logics is undecidable. We mention some of the future works possible in the Conclusion.

    \chapter{Preliminaries}

In this chapter, we first formally define the syntax and semantics of constraint LTL (\cite{DD07}) and temporal logic of repeating values (\cite{DDG12}). Then we shall be stating some important results related to the satisfiability problems for these logics which have been proved in previous papers. This shall be useful in the subsequent chapters where we shall be considering extensions of these logics and their satisfiability problems. The decidability/undecidability results that we shall be proving for these extensions will be built based on the the results discussed in this chapter.\\

The set of natural numbers, integers, rationals, and reals, will be denoted by $\mathbb{N}$, $\mathbb{Z}$, $\mathbb{Q}$, and $\mathbb{R}$ respectively. An infinite word (or sequence) over an alphabet $\Sigma$ is a function $\alpha : \mathbb{N} \to \Sigma$, written as $\alpha(0)\alpha(1)...$. For an infinite word $\alpha$, we will use $\alpha[i, j]$ to denote the finite word $\alpha(i)\alpha(i + 1)...\alpha(j)$, and $\alpha[i,\infty)$ to denote the infinite suffix $\alpha(i)\alpha(i + 1)...$. The set of all infinite words over $\Sigma$ is denoted by $\Sigma^{\omega}$. For a finite word $\tau$ and a finite or infinite word $\gamma$, we use $\tau.\gamma$ to denote the concatenation of $\tau$ and $\gamma$, and $\tau^{\omega}$ to denote the infinite word $\tau.\tau...$.

\section{Constraint LTL}

This logic is an extension of the classical propositional linear temporal logic. But here, instead of boolean variables, the logic will be defined over data variables whose values could range over an infinite domain. In the syntax of this logic, the propositional letters of LTL will be replaced by constraints over these data variables. 

Before we formally begin defining the syntax and semantics of this logic, let us recall a few terms first (\cite{DD07}).

A constraint system $\mathcal{D}$ is of the form $(D, R_1 . . . , R_n, \mathcal{I})$, where $D$ is a non-empty set referred to as the domain, and each $R_i$ is a predicate symbol of arity $a_i$ , with $\mathcal{I}(Ri) \subseteq D^{a_i}$ being its interpretation. We will suppress the mention of $\mathcal{I}$ whenever it is clear from the context. Let us fix such a constraint system $\mathcal{D}$ for the rest of this section.\\

An (atomic) $\mathcal{D}$-constraint over a finite set of variables $U$ is of the form $R(x_1, . . . , x_a)$, where $R$ is a predicate symbol of arity $a$ in $\mathcal{D}$ and each $x_i \in U$. A $D$-valuation over $U$ is a map $s : U \to D$. We will use the notation $val_D(U)$ to denote the set of $D$-valuations over $U$. Let $c = R(x_1, . . . , x_a)$ be a $\mathcal{D}$-constraint over $U$, and let $s$ be a $D$-valuation over $U$. Then we say $s$ satisfies $c$, written $s \models_{D} c$ iff $(s(x1), . . . , s(xa)) \in \mathcal{I}(R)$.\\

Formally, an (atomic) $\mathcal{D}$-term constraint over a set of variables $U$, is of the form
$R(X^{n_1}x_1, . . . ,X^{n_a}x_a)$ where $x_1, . . . , x_a \in U$, and $n_1, . . . , n_a \in \mathbb{N}$. Here $X^{i}$ stands for the juxtaposition of the next-state operator $X$ $i$ times, with $X^{0}x$ representing just $x$. A $\mathcal{D}$-term constraint is interpreted over a sequence of $D$-valuations. A $D$-valuation sequence over a set of variables $U$ is an infinite sequence $\sigma$ of $D$-valuations over U. We say a $D$-valuation sequence $\sigma$ over $U$ satisfies a $\mathcal{D}$-term constraint $c = R(X^{n_1}x_1, . . . ,X^{n_a}x_a)$ over $U$, written $\sigma \models_{D} c$ iff $(\sigma(n_1)(x_1), . . . , \sigma(n_a)(x_a)) \in \mathcal{I}(R)$.

We now introduce CLTL($\mathcal{D}$), the constraint Linear-Time Temporal Logic parametrised by the
constraint system $\mathcal{D}$. Let Var be a countably infinite set of variables, which we fix for the rest of this thesis. The syntax of CLTL($\mathcal{D}$) is given by:

\begin{equation*}
    \phi ::= c|\> \phi \vee \phi \>|\> \neg \phi \>|\> \textnormal{\textbf{X}} \phi \>|\> \phi\textnormal{\textbf{U}}\phi
\end{equation*}

where $c$ is any $\mathcal{D}$-term constraint over the set of variables $Var$.\\

Models for CLTL($\mathcal{D}$) formulas are $D$-valuation sequences over the variables $Var$.  Let $\phi$ be a CLTL($\mathcal{D}$) formula, and $\sigma$ be D-valuation sequence over $Var$. The satisfaction relation $\models$ is defined inductively below:

\begin{itemize}
    \item $\sigma, i \models c$ iff $\sigma[i, \infty) \models_{D} c$
    
    \item $\sigma, i \models \phi \vee \phi^{\prime}$ iff $\sigma, i \models \phi$ or $\sigma, i \models \phi^{\prime}$
    
    \item $\sigma, i \models \neg \phi$ iff $\sigma, i \not \models  \phi$
    
    \item $\sigma, i \models \textnormal{\textbf{X}} \phi$ iff $\sigma, i+1 \models \phi$
    
    \item $\sigma, i \models \phi\textnormal{\textbf{U}}\phi^{\prime}$ iff there is $i \leq j$  such that $\sigma, j \models \phi^{\prime}$ and for every $i \leq l < j$, we have $\sigma, l \models \phi$
    
\end{itemize}

We write $\sigma \models \phi$ if $\sigma, 0 \models \phi$. We shall use the standard derived temporal operators (\textbf{G}, \textbf{F}), and derived Boolean operators$ (\vee, \Rightarrow)$.

A constraint system $\mathcal{D}$ is said to satisfy the completion property if, essentially, given a consistent set of constraints $X$ over a set of variables $U$, any partial valuation which respects the constraints in $X$ involving only the assigned variables, can be extended to a valuation which respects all the constraints in $X$ .

The constraint system $(\mathbb{R},<, \approx)$ is one example of a system that satisfies the completion property.

An example of a constraint system which does not satisfy the completion property is $(\mathbb{Z},<,\approx)$, since for the set of constraints $X = \{x < y,x < z,z < y, x = x, y = y, z = z\}$ over the set of variables $U = \{x, y, z\}$, the partial valuation $s : x \mapsto 0, y \mapsto 1$ satisfies the constraints in $X$ involving $x$ and $y$, but cannot be extended to a valuation which satisfies the constraints $x < z$ and $z < y$ in $X$.

Let $\mathcal{D} = (D, R_1, . . . , R_n)$ be a constraint system. The satisfiability problem for CLTL($\mathcal{D}$) is: Given a CLTL($\mathcal{D}$) formula $\phi$, does there exist a $D$-valuation sequence which satisfies $\phi$? We consider this question for the specific case of constraint systems of the form $\mathcal{D} = (D, <, \approx)$. \\

We say a constraint system, $\mathcal{D} = (D, <, \approx)$ is dense (with respect to the ordering $<$) if for each $d, d^{\prime \prime}$ in $D$ with $d < d^{\prime \prime}$ , there exists $d^{\prime}$ in $D$, such that $d < d^{\prime} < d^{\prime \prime}$. We say $D$ is open iff for each $d^{\prime} \in D$, there exist $d, d^{\prime \prime} \in D$ with $d < d^{\prime} < d^{\prime \prime}$.\\

A frame over a set of variables $U$, w.r.t. a constraint system $\mathcal{D}$, is essentially a maximally consistent set of $\mathcal{D}$-constraints over $U$. More precisely, let $frame_{\mathcal{D}}(s, U)$ denote the (possibly empty) set of $\mathcal{D}$-constraints over $U$ satisfied by a $D$-valuation $s$ over $U$. Then a set of $\mathcal{D}$-constraints $X$ is a frame over the variables $U$, w.r.t. $\mathcal{D}$, if there exists a $D$-valuation $s$ over $U$ such that $X = frame_{\mathcal{D}}(s, U)$.

The frame checking problem for a constraint system $\mathcal{D}$ is as follows: Given a set of $\mathcal{D}$-constraints $X$ , and a finite set of variables $U$, is $X$ a frame over $U$, w.r.t. $\mathcal{D}$?

Now, we have the following theorems from \cite{DD07}:

\begin{lemma}
If a constraint system $\mathcal{D}$ satisfies the completion property and allows frame checking in PSPACE, then the satisfiability problem for CLTL($\mathcal{D}$) is in PSPACE. (\cite[Theorem 4.4]{DD07})
\end{lemma}

\begin{lemma}
Let $\mathcal{D}$ be a constraint system of the form $\mathcal{D} = (D, <, \approx)$ where D is infinite and $<$ is a total order. Then, $\mathcal{D}$ satisfies completion iff $\mathcal{D}$ is dense and open. (\cite[Lemma 5.3]{DD07})
\end{lemma}

From these, we can conclude the following:
\begin{theorem}
Let $\mathcal{D}$ be a constraint system of the form $\mathcal{D} = (D, <, \approx)$ where D is infinite and $<$ is a total order. Let $\mathcal{D}$ be dense and open. Then the satisfiability problem for CLTL($\mathcal{D}$) is decidable. In particular, the satisfiability problem for CLTL($(\mathbb{R}, <, \approx)$) is decidable.
\end{theorem}

\section{Temporal Logic of Repeating Values}

In this section, we consider yet another extension of the classical propositional linear temporal logic. This extension enables one to reason about repetitions of data values from an infinite domain. As before, the logic is defined over a set of data variables whose values can range over an infinite set. Here, the only predicate that we are endowed with, is the equality predicate. But the syntax of this extension allows one to check if the value of a variable has been repeated somewhere in the future. Let us make it more precise. (\cite[Section 2]{DDG12}).

Let $Var = {x_1, x_2, . . .}$ be a countably infinite set of variables. We denote by CLTL$^{\text{XF}}$, the logic whose formulae are defined as follows: 

\begin{equation*}
    \phi ::= x \approx \textnormal{\textbf{X}}^{i} y \>|\> x \approx \textnormal{\textbf{XF}} y \>|\> \phi \wedge \phi \>|\> \neg \phi \>|\> \textnormal{\textbf{X}} \phi \>|\> \phi\textnormal{\textbf{U}}\phi \>|\> \textnormal{\textbf{X}}^{-1} \phi \>|\> \phi\textnormal{\textbf{S}}\phi
\end{equation*}

where $x, y \in Var$, and $i \in \mathbb{N}$. Formulae of the form either $x \approx \textnormal{\textbf{X}}^{i} y$ or $x \approx \textnormal{\textbf{XF}} y$ are said to be atomic and an expression of the form $\textnormal{\textbf{X}}^{i} x$ (abbreviation for $i$ next symbols followed by a variable) is called a term.

A valuation is defined to be a map from $Var$ to $\mathbb{N}$. A CLTL$^{\text{XF}}$ model is an infinite sequence $\sigma$ of valuations. All the subsequent developments
can be equivalently done with the domain $\mathbb{N}$ replaced by an infinite set $D$ since only equality tests are performed in the logics. For every model $\sigma$ and $i \geq 0$, the satisfaction relation $\models$ is defined inductively as follows:

\begin{itemize}
    \item $\sigma, i \models x \approx \textnormal{\textbf{X}}^{j} y$ iff $\sigma(i)(x) = \sigma(i + j)(y)$
    
    \item $\sigma, i \models x \approx \textnormal{\textbf{XF}} y$ iff there exists $j$ such that $i < j$ and $\sigma(i)(x) = \sigma(j)(y)$
    
    \item $\sigma, i \models \phi \wedge \phi^{\prime}$ iff $\sigma, i \models \phi$ and $\sigma, i \models \phi^{\prime}$
    
    \item $\sigma, i \models \neg \phi$ iff $\sigma, i \not \models  \phi$
    
    \item $\sigma, i \models \textnormal{\textbf{X}} \phi$ iff $\sigma, i+1 \models \phi$
    
    \item $\sigma, i \models \phi\textnormal{\textbf{U}}\phi^{\prime}$ iff there is $i \leq j$  such that $\sigma, j \models \phi^{\prime}$ and for every $i \leq l < j$, we have $\sigma, l \models \phi$
    
    \item $\sigma, i \models \textnormal{\textbf{X}}^{-1} \phi$ iff $i \geq 1$ and $\sigma, i-1 \models \phi$
    
    \item $\sigma, i \models \phi\textnormal{\textbf{S}}\phi^{\prime}$ iff there is $j \leq i$  such that $\sigma, j \models \phi^{\prime}$ and for every $j < l \leq i$, we have $\sigma, l \models \phi$
\end{itemize}

We write $\sigma \models \phi$ if $\sigma, 0 \models \phi$. We shall use the standard derived temporal operators (\textbf{G}, \textbf{F}), and derived Boolean operators$ (\vee, \Rightarrow)$. We also use the notation $\textnormal{\textbf{X}}^{i}x \approx \textnormal{\textbf{X}}^{j} y$ as an abbreviation for the formula $\textnormal{\textbf{X}}^{i}(x \approx \textnormal{\textbf{X}}^{j-i} y)$ (assuming without any loss of generality that $i \leq j$).

The satisfiability problem for CLTL$^{\text{XF}}$ is to check for a given CLTL$^{\text{XF}}$ formula $\phi$, whether there exists a model $\sigma$ such that $\sigma \models \phi$. 

Now we consider an extension CLTL$^{\textnormal{XF},\textnormal{XF}^{-1}}$  (section 9, \cite{DDG12}) be the extension of CLTL$^{\text{XF}}$ with atomic formulae of the form $x \approx \textnormal{\textbf{XF}}^{-1} y$. The satisfaction relation is extended as follows:

\begin{itemize}
    \item $\sigma, i \models x \approx \textnormal{\textbf{XF}}^{-1} y$ iff there exists $j \geq 0$ such that $j < i$ and $\sigma(i)(x) = \sigma(j)(y)$

\end{itemize}

The satisfiability problem for CLTL$^{\textnormal{XF},\textnormal{XF}^{-1}}$ is to check for a given CLTL$^{\textnormal{XF},\textnormal{XF}^{-1}}$ formula $\phi$, whether there exists a model $\sigma$ such that $\sigma \models \phi$. \\

Now we have the following results from \cite{DDG12}:

\begin{theorem}
The satisfiability problem for CLTL$^{\text{XF}}$ is decidable. (\cite[Theorem 4]{DDG12})
\end{theorem}

\begin{theorem}
The satisfiability problem for CLTL$^{\textnormal{XF},\textnormal{XF}^{-1}}$ is decidable. (\cite[Theorem 6(I)]{DDG12})
\end{theorem}

    \chapter{Decidability of an extension of CLTL$^{\textbf{XF}}$}
\section{Introduction}
Formulae of the logic $\text{CLTL}^{\text{XF}}$ (described in the previous chapter) introduced in \cite{DDG12} are defined by the grammar
below:
\begin{equation*}
    \phi ::= x \approx \textnormal{\textbf{X}}^{i} y \>|\> x \approx \textnormal{\textbf{XF}} y \>|\> \phi \wedge \phi \>|\> \neg \phi \>|\> \textnormal{\textbf{X}} \phi \>|\> \phi\textnormal{\textbf{U}}\phi \>|\> \textnormal{\textbf{X}}^{-1} \phi \>|\> \phi\textnormal{\textbf{S}}\phi
\end{equation*}

Let $V = V_{\text{local}} \biguplus V_{\text{remote}}$ be the set of variables. Let us consider a fragment of LTL with data comparisons, which extends
$\text{CLTL}^{\text{XF}}$ by addition of local constraints with $<$ or $>$ with the following restrictions:
\begin{itemize}
    \item atomic formulae of the form $x \approx \textnormal{\textbf{XF}} y$ are allowed only if $x, y \in V_{\text{remote}}$.
    \item atomic formulae of the form $x \approx \textnormal{\textbf{X}}^{i} y$ are allowed if $\{x, y\} \subseteq V_{\text{local}}$ or if $\{x, y\} \subseteq V_{\text{remote}}$.
   \item atomic formulae of the form $x < \textnormal{\textbf{X}}^{i} y$ are allowed only if $x, y \in V_{\text{local}}$.
\end{itemize}
The formal syntax of this extension is given by the following grammar : 
\begin{align*}
     \phi :: = \hspace{0.06 cm} & x_{\text{loc}} \approx \textnormal{\textbf{X}}^{i} y_{\text{loc}} \hspace{0.06cm} |\hspace{0.06cm} x_{\text{loc}} < \textnormal{\textbf{X}}^{i} y_{\text{loc}} \hspace{0.06cm}|\hspace{0.06cm} x_{\textnormal{rem}} \approx \textnormal{\textbf{X}}^{i} y_{\textnormal{rem}\>} \hspace{0.06cm} | \hspace{0.06cm} x_{\text{rem}} \approx \textnormal{\textbf{XF}} y_{\text{rem}} \hspace{0.06cm}  \\ 
     &| \hspace{0.06cm} \phi \wedge \phi \hspace{0.06cm}
     | \hspace{0.06cm} \neg \phi \hspace{0.06cm} | \hspace{0.06cm} \textnormal{\textbf{X}} \phi \hspace{0.06cm} | \hspace{0.06cm} \phi\textnormal{\textbf{U}}\phi \hspace{0.06cm} | \hspace{0.06cm} \textnormal{\textbf{X}}^{-1} \phi 
     |\hspace{0.06cm} \phi\textnormal{\textbf{S}}\phi, \\ & \text{ where } x_{\text{loc}}, y_{\text{loc}} \in V_{\text{local}},\hspace{0.12cm} x_{\text{rem}}, y_{\text{rem}} \in V_{\text{remote}}
\end{align*}
We call this extension CLRV$^{\textnormal{XF}}$($\mathcal{D}$) where the domain $\mathcal{D}$ is of the form $(D, <, \approx)$, where $D$ is an infinite set, `$<$' is any strict total ordering over $D$ and `$\approx$' is the equality relation.\\

A $D$-valuation is defined to be a map from $V$ to $D$. A $D_{\text{loc}}$-valuation is defined to be a map from $V_{\text{local}}$ to $D$. 
Similarly, a $D_{\text{rem}}$-valuation is defined to be a map from $V_{\text{remote}}$ to $D$. A local model is an infinite sequence $\sigma_l$ of $D_{\text{loc}}$-valuations and a remote model is an infinite sequence $\sigma_r$ of $D_{\text{rem}}$-valuations.
A model is an infinite sequence $\sigma$ of $D$-valuations. 

We write $\sigma \models \phi$ if $\sigma, 0 \models \phi$. We shall use the standard derived temporal operators (\textbf{G}, \textbf{F}), and derived Boolean operators$ (\vee, \Rightarrow)$. We also use the
notation $\textnormal{\textbf{X}}^{i}x \approx \textnormal{\textbf{X}}^{j} y$ as an abbreviation for the formula $\textnormal{\textbf{X}}^{i}(x \approx \textnormal{\textbf{X}}^{j-i} y)$ (assuming without any
loss of generality that $i \leq j$).\\

Now we consider the problem of satisfiability of this extension over dense and open domains (in particular, over $\mathbb{R}$)\\

\begin{problem}
Let $\mathcal{D} = (D, <, \approx)$ be a constraint system which is open and dense (defined in the previous chapter). The satisfiability problem for CLRV$^{\textnormal{XF}}$($\mathcal{D}$) is: Given a CLRV$^{\textnormal{XF}}$($\mathcal{D}$) formula $\phi$, does there exist a model which satisfies $\phi$?
\end{problem}

We shall now prove that the satisfiability problem for CLRV$^{\textnormal{XF}}$($\mathcal{D}$) is decidable. Let us begin by defining some notations.\\

Let $l$ be the maximal $i$ such that a term of the form $\textnormal{\textbf{X}}^{i} x_{\textnormal{loc}}$ occurs in $\phi$ (where $x_{\textnormal{loc}} \in V_{\textnormal{local}}$). The value $l$ is called the local \textbf{X}-length of $\phi$. \\

Let $r$ be the maximal $i$ such that a term of the form $\textnormal{\textbf{X}}^{i} x_{\textnormal{rem}}$ occurs in $\phi$ (where $x_{\textnormal{rem}} \in V_{\textnormal{remote}}$). The value $r$ is called the remote \textbf{X}-length of $\phi$. \\

We use $\Omega^{l}$ to denote the set of all constraints of the form $\textnormal{\textbf{X}}^{i} x_{\text{loc}} \approx \textnormal{\textbf{X}}^{j} y_{\text{loc}}$ or $\textnormal{\textbf{X}}^{i} x_{\text{loc}} < \textnormal{\textbf{X}}^{j} y_{\text{loc}}$, with $x_{\text{loc}} \in V_{\text{local}}$ and $i, j \in \{0,1,...l\}$. Similarly, $\Omega^{r}$ denotes the set of all constraints of the form $\textnormal{\textbf{X}}^{i} x_{\textnormal{rem}} \approx \textnormal{\textbf{X}}^{j} y_{\textnormal{rem}}$ or $\textnormal{\textbf{X}}^{i}(x_{\text{rem}} \approx \textnormal{\textbf{XF}} y_{\text{rem}})$, with $x_{\text{rem}} \in V_{\text{remote}}$ and $i, j \in \{0,1,...r\}$\\

Let $\sigma_l$ be a local $D_{\text{loc}}$-valuation sequence and let $l \in \mathbb{N}$. We define, $l$-$frame_{\mathcal{D}}(\sigma_l) = \{c \in \Omega^{l} \mid \sigma_l \models c\}$. An $l$-$frame$ with respect to $\mathcal{D}$ is a subset $FR^{l}$ of $\Omega^{l}$ such that $FR^{l} =$
$l$-$frame_{\mathcal{D}}(\sigma_l)$ for some $D_{\text{loc}}$-valuation sequence $\sigma_l$ (same as the definition of $k$-frame of \cite{DD07}) . Let the set of all $l$-frames be denoted by $Frame^l$.\\

An $r$-$frame$ is defined to be a set of constraints $FR^{r} \subseteq \Omega^{r}$ which satisfies the following constraints (same as the definition of $(l,k)$-frame of \cite{DDG12}) :
\begin{enumerate}
    \item For all $i \in \{0, . . . , r\}$ and $x \in V_{\textnormal{remote}}$, $\textnormal{\textbf{X}}^{i} x \approx \textnormal{\textbf{X}}^{i} x \in FR^{r}$ .
    \item
For all $i, j \in \{0, . . . , r\}$ and $x, y \in V_{\textnormal{remote}}$, $\textnormal{\textbf{X}}^{i} x \approx \textnormal{\textbf{X}}^{j} y \in FR^{r}$ iff $\textnormal{\textbf{X}}^{j} y \approx \textnormal{\textbf{X}}^{i} x \in FR^{r}$
\item
For all $i, j, j' \in \{0, . . . , r\}$ and $x, y, z \in V_{\textnormal{remote}}$, if $\{\textnormal{\textbf{X}}^{i} x \approx \textnormal{\textbf{X}}^{j} y$, $\textnormal{\textbf{X}}^{j} y \approx \textnormal{\textbf{X}}^{j'} z\} \subseteq FR^{r}$ then
$\textnormal{\textbf{X}}^{i} x \approx \textnormal{\textbf{X}}^{j'} z \in FR^{r}$
\item
For all $i, j \in \{0, . . . , r\}$ and $x, y \in V_{\textnormal{remote}}$ such that $\textnormal{\textbf{X}}^{i} x \approx \textnormal{\textbf{X}}^{j} y \in FR^{r}$ :
\begin{itemize}
    \item if $i = j$, then for every $z \in V_{\textnormal{remote}}$ we have $\textnormal{\textbf{X}}^{i}(x \approx \textnormal{\textbf{XF}}z) \in FR^{r}$ iff $\textnormal{\textbf{X}}^{j}(y \approx \textnormal{\textbf{XF}} z) \in FR^{r}$ ;
    
    \item if $i < j$ then $\textnormal{\textbf{X}}^{i}(x \approx \textnormal{\textbf{XF}}y) \in FR^{r}$ , and for $z \in V_{remote}$, $\textnormal{\textbf{X}}^{i}(x \approx \textnormal{\textbf{XF}}z) \in FR^{r}$ iff either
$\textnormal{\textbf{X}}^{j}(y \approx XFz) \in FR^{r}$ or there exists $i < j' \leq j$ such that $\textnormal{\textbf{X}}^{i} x \approx \textnormal{\textbf{X}}^{j'} z \in FR^{r}$ .
    
\end{itemize}
\end{enumerate}
Let the set of all $r$-frames be denoted by $Frame^{r}$. \\

We now define an $(l, r)$-$frame$ to be a set of constraints $FR \subseteq \Omega^{l} \cup \Omega^{r}$ of the form $FR^{l} \cup FR^{r}$ where $FR^{l}$ is an $l$-$frame$ and $FR^{r}$ is an $r$-$frame$. Let the set of all $(l, r)$-frames be denoted by $Frame^{(l,r)}$.\\

The local projection $FR^l$ of an $(l, r)$-frame $FR$ is defined as $FR^l = FR \cap \Omega^{l}$. The remote projection $FR^r$ of an $(l, r)$-frame $FR$ is defined as $FR^r = FR \cap \Omega^{r}$.\\

A pair of $l$-frames $(FR_1, FR_2)$ is one-step consistent (same as the definition of locally consistent $k$-frames of \cite{DD07}), if it satisfies the following conditions:
\begin{itemize}
    \item For all $\textnormal{\textbf{X}}^{i} x \approx \textnormal{\textbf{X}}^{j} y \in \Omega^{l}$
with $0 < i, j$, we have $\textnormal{\textbf{X}}^{i} x \approx \textnormal{\textbf{X}}^{j} y \in FR_1$ iff $\textnormal{\textbf{X}}^{i-1} x \approx \textnormal{\textbf{X}}^{j-1} y \in FR_2$,
    \item For all $\textnormal{\textbf{X}}^{i} x < \textnormal{\textbf{X}}^{j} y \in \Omega^{l}$
with $0 < i, j$, we have $\textnormal{\textbf{X}}^{i} x < \textnormal{\textbf{X}}^{j} y \in FR_1$ iff $\textnormal{\textbf{X}}^{i-1} x < \textnormal{\textbf{X}}^{j-1} y \in FR_2$,
\end{itemize}
A pair of $r$-frames $(FR_1, FR_2)$ is one-step consistent \cite{DDG12}, if it satisfies the following conditions:
\begin{itemize}
\item For all $\textnormal{\textbf{X}}^{i} x \approx \textnormal{\textbf{X}}^{j} y \in \Omega^{r}$
with $0 < i, j$, we have $\textnormal{\textbf{X}}^{i} x \approx \textnormal{\textbf{X}}^{j} y \in FR_1$ iff $\textnormal{\textbf{X}}^{i-1} x \approx \textnormal{\textbf{X}}^{j-1} y \in FR_2$,
\item For all $\textnormal{\textbf{X}}^{i}(x \approx \textnormal{\textbf{XF}}y) \in \Omega^{r}$
with $i > 0$, we have $\textnormal{\textbf{X}}^{i}(x \approx \textnormal{\textbf{XF}}y) \in FR_1$ iff $\textnormal{\textbf{X}}^{i-1}(x \approx \textnormal{\textbf{XF}}y) \in FR_2$.
\end{itemize}
A pair of $(l, r)$-frames $(FR_1, FR_2)$ is said to be one-step consistent, if the pair of local projections of $FR_1$ and $FR_2$ is one-step consistent and the pair of remote projections of $FR_1$ and $FR_2$ is also one-step consistent. \\

An $l$-symbolic model is an infinite sequence $\rho$ of $l$-frames such that for all $i \in \mathbb{N}$, the pairs $(\rho(i), \rho(i+1))$ are one-step consistent. An $r$-symbolic model is an infinite sequence $\rho$ of $r$-frames such that for all $i \in \mathbb{N}$, the pairs $(\rho(i), \rho(i+1))$ are one-step consistent. 
An $(l, r)$-symbolic model (or a symbolic model) is an infinite sequence $\rho$ of $(l, r)$-frames such that for all $i \in \mathbb{N}$, the pairs $(\rho(i), \rho(i+1))$ are one-step consistent. \\

The local projection of an $(l, r)$-symbolic model $\rho$ is the infinite sequence $\rho_l$ such that $\rho_l (i)$ is the local projection of $\rho(i)$ for all $i \in \mathbb{N}$. Similarly, the remote projection of an $(l, r)$-symbolic model $\rho$ is the infinite sequence $\rho_r$ such that $\rho_r (i)$ is the remote projection of $\rho(i)$ for all $i \in \mathbb{N}$. \\

We say a model $\sigma$ realizes a symbolic model $\rho$ (or equivalently that $\rho$ admits a model $\sigma$) iff for every $i \in \mathbb{N}$, we have $\rho(i) = \{\phi \in (\Omega^l \cup \Omega^r) \mid \sigma, i \models \phi\}$. We say a local model $\sigma_l$ realizes the local projection $\rho_l$ of a symbolic model $\rho$ (or equivalently that the local projection of $\rho$ admits a local model $\sigma_l$) iff for every $i \in \mathbb{N}$, we have $\rho_l(i) = \{\phi \in \Omega^l \mid \sigma, i \models \phi\}$. Similarly, a remote model $\sigma_r$ realizes the remote projection $\rho_r$ of a symbolic model $\rho$ (or equivalently the remote projection of $\rho$ admits a local model $\sigma_r$) iff for every $i \in \mathbb{N}$, we have $\rho_r(i) = \{\phi \in \Omega^r \mid \sigma, i \models \phi\}$.\\

It is easy to observe that the following lemma holds:
\begin{lemma}
A symbolic model $\rho$ admits a model $\sigma$ iff the local projection $\rho_l$ of $\rho$ admits a local model and the remote projection $\rho_r$ of $\rho$ admits a remote model.
\end{lemma} 
Let us define the symbolic satisfaction relation $\rho, i \models_{symb} \phi$, where $\phi$ is a formula of local \textbf{X}-length at most $l$ and remote \textbf{X}-length at most $r$ and $\rho$ is an $(l,r)$-symbolic model. The relation $\models_{symb}$ is defined in the same way as $\models$ for CLRV$^{\textnormal{XF}}$($\mathcal{D}$), except that, for atomic formulae $\phi$ we have :
\begin {equation*}
 \rho, i \models_{symb} \phi \Leftrightarrow \phi \in \rho(i)
\end{equation*}
Now we have the following lemma:
\begin{lemma}
A CLRV$^{\textnormal{XF}}$($\mathcal{D}$) formula $\phi$ of local \textbf{X}-length $l$ and remote \textbf{X}-length $r$ over $V = V_{\text{local}} \biguplus V_{\text{remote}}$ is satisfiable iff there exists an infinite symbolic model $\rho$ such that $\rho \models_{symb} \phi$ and $\rho$ admits a model.
\end{lemma}
\begin{proof}
The lemma clearly follows from Lemma 3 of \cite{DDG12} and Corollary 4.1 of \cite{DD07}.
\end{proof}
\section{Automata-based decidability result}
We shall now give an automata-based decision procedure for satisfiability of CLRV$^{\textnormal{XF}}$($\mathcal{D}$). This decision procedure has been adapted from the automata-based decision procedures of previous papers (refer section 4-5 of \cite{DD07}, section 4-7 of \cite{DDG12}).
Given a CLRV$^{\textnormal{XF}}$($\mathcal{D}$) formula $\phi$, we will construct an automaton $\mathcal{A}_{\phi}$ whose language is nonempty iff $\phi$ is satisfiable. The automaton $\mathcal{A}_{\phi}$ will be a special kind of counter automaton with a generalized B{\"u}chi acceptance condition, whose non-emptiness problem is decidable.\\

Let $\phi$ be a CLRV$^{\textnormal{XF}}$($\mathcal{D}$) formula of local \textbf{X}-length $l$ and remote \textbf{X}-length $r$ over $V = V_{\text{local}} \biguplus V_{\text{remote}}$. We define $\mathcal{A}_{\phi}$ as an automaton over the alphabet $Frame^{(l,r)}$ that accepts the intersection of the languages accepted by the three automata $\mathcal{A}_{1sc}$, $\mathcal{A}_{symb}$ and $\mathcal{A}_{realizable}$ described below:
\begin{itemize}
    \item $\mathcal{A}_{1sc}$ recognizes the set of ``valid" symbolic models (i.e., sequences of frames such that every pair of consecutive frames is one-step consistent),
    \item $\mathcal{A}_{symb}$ recognizes the set of symbolic models satisfying $\phi$,
    \item $\mathcal{A}_{realizable}$ recognizes the set of symbolic models that are realizable.
\end{itemize}
The automaton $\mathcal{A}_{1sc}$ is a B{\"u}chi automaton that checks that the sequence is one-step consistent. $\mathcal{A}_{1sc} = (Q, Q_0, F, \rightarrow)$ such that 
\begin{itemize}
    \item $Q$ is the set of all $(l,r)$-frames and $Q_0 = Q$,
    \item the transition relation is defined by $FR_1 \xrightarrow{FR} FR_2$ iff $FR = FR_1$ and the pair $(FR_1, FR_2)$ is one-step consistent,
    \item the set of final states F is equal to Q.
\end{itemize}
Note that $\phi$ can be thought of as an LTL formula over $\Omega^l \cup \Omega^r$. We define $\mathcal{A}_{symb}$ to be the Vardi-Wolper automaton (refer \cite{VW86}) for the LTL formula $\phi$. It is clear that $\mathcal{A}_{symb}$ accepts exactly the symbolic models satisfying $\phi$.\\

Now one can notice that the local projections of the symbolic models recognized by the automaton $\mathcal{A}_{1sc}$ are also one-step consistent. Now, a constraint system $\mathcal{D}$ is said to satisfy the completion property if, essentially, given a consistent set of constraints $X$ over a set of variables $U$, any partial valuation which respects the constraints in $X$ involving only the assigned variables, can be extended to a valuation which respects all the constraints in $X$ . We have the following lemmas:
\begin{lemma}
Let $\mathcal{D}$ be a constraint system which satisfies the completion property. Then every one-step
consistent l-frame sequence w.r.t. $\mathcal{D}$ admits a D-valuation sequence. (\cite[Lemma 4.2]{DD07}) 
\end{lemma}
\begin{lemma}
Let $\mathcal{D}$ be a constraint system of the form $\mathcal{D} = (D, <, \approx)$ where D is infinite and $<$ is a total order. Then, $\mathcal{D}$ satisfies completion iff $\mathcal{D}$ is dense and open. (\cite[Lemma 5.3]{DD07})
\end{lemma}
Now using Lemmas 2.2.1 and 2.2.2, we get the following lemma:
\begin{lemma}
The local projection $\rho_l$ of every symbolic model $\rho$, recognized by $\mathcal{A}_{1sc}$, admits a local model. 
\end{lemma}
Now from \cite{DDG12}, we have the following lemma:
\begin{lemma}
There exists a simple counter automaton $\mathcal{A}_{real}$ (over the alphabet $Frame^r$) which recognizes the set of all realizable $r$-symbolic models. (\cite[section 7]{DDG12})
\end{lemma}
Now we state the following important lemma:
\begin {lemma}
There exists a simple counter automaton $\mathcal{A}_{realizable}$ which accepts exactly those symbolic models whose remote projections are realizable i.e., admit a remote model. 
\end{lemma}
\begin{proof}
From the previous lemma (Lemma 2.2.4), we know that it is possible to construct an automaton $\mathcal{A}_{real}$ over the alphabet $Frame^r$ which accepts the set of all realizable $r$-symbolic models. The construction of this automaton has been described in detail in Section 7 of \cite{DDG12}. We modify that construction to construct a simple counter automaton $\mathcal{A}_{realizable}$ that accepts exactly those symbolic models whose remote projections are realizable. Basically, we construct the automaton $\mathcal{A}_{realizable}$ over the alphabet $Frame^{(l,r)}$ and at each transition which is described in the construction of $\mathcal{A}_{real}$, we make the corresponding check on the remote projection of the frame that we are reading at that transition (instead of making the check on the frame itself).
\end{proof}
It is known that:
\begin{lemma}
If A is a simple counter automaton, and B is a B{\"u}chi automaton, then we can
construct a simple counter automaton accepting the language $L(A) \cap L(B)$. (\cite[Proposition 1]{DDG12})
\end{lemma}
Now from Lemma 2.2.3, Lemma 2.2.5, Lemma 2.2.6 and using Lemma 2.1.1, we get:
\begin{lemma}
$\mathcal{A}_\phi = \mathcal{A}_{1sc} \cap \mathcal{A}_{symb} \cap \mathcal{A}_{realizable}$ is a simple counter automaton that accepts only ``valid" symbolic models of the form $\rho$ such that $\rho$ admits a model and $\rho \models_{symb} \phi$.
\end{lemma}
Using Lemma 2.1.2 and Lemma 2.2.7, we get the following corollary:
\begin{corollary}
A CLRV$^{\textnormal{XF}}$($\mathcal{D}$) formula $\phi$ of local \textbf{X}-length $l$ and remote \textbf{X}-length $r$ over $V = V_{\text{local}} \biguplus V_{\text{remote}}$ is satisfiable iff the language of the corresponding simple counter automaton $\mathcal{A}_\phi$, is non-empty.
\end{corollary}
We know that the non-emptiness problem for simple counter automata is decidable (\cite[Lemma 6]{DDG12}).
Hence, we have the following theorem:
\begin{theorem}
The satisfiability problem for CLRV$^{\textnormal{XF}}$($\mathcal{D}$) is decidable. 
\end{theorem}
Note that $\mathcal{D}$ was any dense and open domain. In particular, $(\mathbb{R}, <, \approx)$ is a dense and open domain. Hence, the satisfiability problem for CLRV$^{\textnormal{XF}}$$((\mathbb{R}, <, \approx))$ is decidable.

    \chapter{Undecidability results}

\section{Logic of Remote Constraints}
\subsection{Introduction}
Formulae of the logic $\text{CLTL}^{\text{XF}, \text{XF}^{-1}}$ (described in chapter 1) introduced in \cite{DDG12} are defined by the grammar
below:
\begin{equation*}
    \phi ::= x \approx \textnormal{\textbf{X}}^{i} y \>|\> x \approx \textnormal{\textbf{XF}} y \>|\> x \approx \textnormal{\textbf{XF}}^{-1} y \>|\> \phi \wedge \phi \>|\> \neg \phi \>|\> \textnormal{\textbf{X}} \phi \>|\> \phi\textnormal{\textbf{U}}\phi \>|\> \textnormal{\textbf{X}}^{-1} \phi \>|\> \phi\textnormal{\textbf{S}}\phi
\end{equation*}

Let us consider a fragment of LTL with data comparisons which extends $\text{CLTL}^{\text{XF}, \text{XF}^{-1}}$ by addition of the constraints: $x < \textnormal{\textbf{XF}} y$ and $x > \textnormal{\textbf{XF}} y$. We consider this extension without the past temporal operators ($\textnormal{\textbf{X}}^{-1}$ and \textbf{S}).\\

The formal syntax of this extension is given by the following grammar :  
\begin{align*}
    \phi ::= \hspace{0.06 cm} & x \approx \textnormal{\textbf{X}}^{i} y \>|\> x \approx \textnormal{\textbf{XF}} y \>|\> x \approx \textnormal{\textbf{XF}}^{-1} y \>|\> x < \textnormal{\textbf{XF}} y \>|\> x > \textnormal{\textbf{XF}} y 
    \> \\
    & |\> \phi \wedge \phi \>|\> \neg \phi \>|\> \textnormal{\textbf{X}} \phi \>|\> \phi\textnormal{\textbf{U}}\phi 
\end{align*}

We call this extension LRC$^{\top}(\mathcal{D})$ where the constraint system $\mathcal{D}$ is of the form $(D, <, \approx)$, where $D$ is an infinite set, `$<$' is any strict total ordering over $D$ and `$\approx$' is the equality relation.\\

A $D$-valuation is defined to be a map from $V$ to $D$. A model is an infinite sequence $\sigma$ of $D$-valuations. For every model $\sigma$ and $i \geq 0$, the satisfaction relation $\models$ is extended as follows:

\begin{itemize}
    
    \item $\sigma, i \models x < \textnormal{\textbf{XF}} y$ iff there exists $j$ such that $i < j$ and $\sigma(i)(x) < \sigma(j)(y)$
    
    \item $\sigma, i \models x > \textnormal{\textbf{XF}} y$ iff there exists $j$ such that $i < j$ and $\sigma(i)(x) > \sigma(j)(y)$

\end{itemize}

The satisfaction relation is defined exactly the same way as in $\text{CLTL}^{\text{XF}, \text{XF}^{-1}}$  for the other formulae.

We write $\sigma \models \phi$ if $\sigma, 0 \models \phi$. We shall use the standard derived temporal operators (\textbf{G}, \textbf{F}), and derived Boolean operators$ (\vee, \Rightarrow)$. We also use the
notation $\textnormal{\textbf{X}}^{i}x \approx \textnormal{\textbf{X}}^{j} y$ as an abbreviation for the formula $\textnormal{\textbf{X}}^{i}(x \approx \textnormal{\textbf{X}}^{j-i} y)$ (assuming without any
loss of generality that $i \leq j$).\\

One could extend LRC$^{\top}(\mathcal{D})$ by replacing $x \approx \textnormal{\textbf{XF}} y$ with $x \approx \langle \phi? \rangle y$ and $x \approx \textnormal{\textbf{XF}}^{-1} y$ with $x \approx {\langle \phi? \rangle}^{-1} y$ in the syntax. We call this extension LRC($\mathcal{D}$). Formally, the syntax of this logic is: 
\begin{align*}
    \phi ::= \hspace{0.06 cm} & x \approx \textnormal{\textbf{X}}^{i} y \>|\> x \approx \langle \phi? \rangle y \>|\> x \approx {\langle \phi? \rangle}^{-1} y \>|\> x < \textnormal{\textbf{XF}} y \>|\> x > \textnormal{\textbf{XF}} y 
    \> \\
    & |\> \phi \wedge \phi \>|\> \neg \phi \>|\> \textnormal{\textbf{X}} \phi \>|\> \phi\textnormal{\textbf{U}}\phi \>|\> \textnormal{\textbf{X}}^{-1} \phi \>|\> \phi\textnormal{\textbf{S}}\phi
\end{align*}
The satisfaction relation for the new atomic formulae is defined as follows:
\begin{itemize}
    \item $\sigma, i \models x \approx \langle \phi? \rangle y$ iff there exists $j$ such that $i < j$, $\sigma(i)(x) = \sigma(j)(y)$ and $\sigma, j \models \phi$
    
    \item $\sigma, i \models x \approx {\langle \phi? \rangle}^{-1} y$ iff there exists $j \geq 0$ such that $j < i$, $\sigma(i)(x) = \sigma(j)(y)$ and $\sigma, j \models \phi$
\end{itemize}
The semantics for the other atomic formulae is the same as in LRC$^{\top}(\mathcal{D})$.\\

It is easy to see that LRC($\mathcal{D}$) is a generalization of LRC$^{\top}(\mathcal{D})$. This is because the formula $x \approx \textnormal{\textbf{XF}} y$ can be expressed as $x \approx \langle \top? \rangle y$ and the formula $x \approx \textnormal{\textbf{XF}}^{-1} y$ can be expressed as $x \approx {\langle \top? \rangle}^{-1} y$.\\

But in fact, it turns out that there is a polynomial-time algorithm, which, given a formula $\phi$ in LRC($\mathcal{D}$), computes a formula $\phi^{\prime}$ in LRC$^{\top}(\mathcal{D})$, that preserves satisfiability: there is a model
$\sigma$ such that $\sigma \models \phi$ iff there is a model $\sigma^{\prime}$ such that $\sigma^{\prime} \models \phi^{\prime}$. The proof follows from the Proof of Corollary 4.6 of \cite{SDM16}. 

\begin{problem}
Let $\mathcal{D} = (D, <, \approx)$ be a constraint system. The satisfiability problem for LRC$^{\top}(\mathcal{D})$ is: Given an LRC$^{\top}(\mathcal{D})$ formula $\phi$, does there exist a model which satisfies $\phi$?
\end{problem}

From the previous section, it is clear that the problem of satisfiability for LRC$^{\top}(\mathcal{D})$ is equivalent to the satisfiability problem for LRC($\mathcal{D}$) \\

We shall now prove that the satisfiability problem for LRC($\mathcal{D}$) is undecidable (and hence that the satisfiability problem for LRC$^{\top}(\mathcal{D})$ is undecidable) 

\subsection{Proof of Undecidability}
The proof is by a reduction from infinite version of Post's Correspondence Problem (PCP). Consider a finite alphabet $\Sigma = \{a_1,a_2,...a_r\}$. An instance of infinite PCP
consists of $k$ pairs $(u_i, v_i)$ of words from $\Sigma^{*}$ and the question is whether there exists an infinite sequence of indices $i_0, i_1, ...$ such that $u_{i_0}u_{i_1}... = v_{i_0}v_{i_1}... $. The infinite version of PCP is known to be undecidable (just like the well-known finite version of PCP).\\

We modify the proof of Proposition 29 of \cite{MCATL10}. There, the authors prove the undecidability of $\textnormal{FO}^2(\sim, \prec, +1, <)$ by giving a reduction from PCP to $\textnormal{FO}^2(\sim, \prec, +1, <)$. Here, we modify the proof and give a reduction from infinite PCP to LRC($\mathcal{D}$) (where $\mathcal{D}$ is any given constraint system).\\ 

Let $\overline{\Sigma} = \Sigma \cup \Sigma^{\prime}$  be the alphabet consisting of two disjoint copies of $\Sigma$. If $w$ is a word in $\Sigma^{\omega}$, then $\overline{w} \in \overline{\Sigma}^{\omega}$ is obtained from $w$ by replacing each letter $a$ with the corresponding letter $\overline{a}$ in $\overline{\Sigma}$. Consider a solution $i_0, i_1, ... $ for the given infinite PCP instance. Let $w$ be the infinite word $u_{i_0}u_{i_1}, ... $, or equivalently the word $v_{i_0}v_{i_1}, ... $. Let $\hat{w}$ be the infinite word $u_{i_0}\overline{v_{i_0}}u_{i_1}\overline{v_{i_1}}...$ over $\overline{\Sigma}$.\\

Now, consider a set of variables $V = \{x, x_1, x_2, ... x_r, y_1, y_2, ... y_r, w_1, w_2\}$ (where $r = |\Sigma|$). Given a solution $\hat{w}$ to an instance of the infinite PCP, we encode it as a $D$-valuation sequence $\sigma$ over $V$, satisfying the following conditions:
\begin{enumerate}[(C1)]

    \item For all $i \geq 0$, $1 \leq j \leq r$:  
    
    $\hat{w}(i) = a_j$ iff $\sigma(i)(x) = \sigma(i)(x_j)$ and $\sigma(i)(x) \not = \sigma(i)(v)$ for any $v \in V\setminus\{x_j\}$;\\
    
    $\hat{w}(i) = \overline{a_j}$ iff $\sigma(i)(x) = \sigma(i)(y_j)$ and $\sigma(i)(x) \not = \sigma(i)(v)$ for any $v \in V\setminus\{y_j\}$
    
    \item The subsequence of values of the variable $x$ at $\Sigma$-positions (i.e. positions $i$ for which $\hat{w}(i) \in \Sigma$) is equal to the subsequence of values of the variable $x$ at $\overline{\Sigma}$-positions (i.e. positions $i$ for which $\hat{w}(i) \in \overline{\Sigma}$) 
    
    \item For all $i \geq 0$, $j \geq 0$, $i \not = j$, if $\hat{w}(i) \in \Sigma$ and $\hat{w}(j) \in \Sigma$, then $\sigma(i)(x) \not = \sigma(j)(x)$. Similarly, for all $i \geq 0$, $j \geq 0$, $i \not = j$, if $\hat{w}(i) \in \overline{\Sigma}$ and $\hat{w}(j) \in \overline{\Sigma}$, then $\sigma(i)(x) \not = \sigma(j)(x)$
    
\end{enumerate}
We now describe an LRC($\mathcal{D}$) formula $\psi$ such that $w$ is a solution to the infinite PCP if and only if $\sigma$ is a model of $\psi$. Note that while describing the formula we shall interchangeably use $(x \approx x_i)$ with $a_i$, $(x \approx y_i)$ with $\overline{a_i}$, $(x \approx w_1)$ with $p$ and $(x \approx w_2)$ with $q$. The formula is the conjunction of the following properties:
\begin{enumerate} [(P1)]
    \item At every position, the value of $x$ is equal to that of exactly one variable in $V \setminus \{x\}$
    
    \item Now that we are taking a conjunction, we can define the notion of a word corresponding to any $D$-valuation sequence satisfying property 1. If $\sigma_0$ is a $D$-valuation sequence satisfying property 1, then the word corresponding to $\sigma_0$ is the infinite word $w_0 \in (\overline{\Sigma} \cup \{p,q\})^\omega$ such that at every position $i$, we have $w_0(i)$ equal to the unique letter from $\overline{\Sigma} \cup \{p,q\}$, that holds at $\sigma_0(i)$. (The existence and uniqueness of the letter that holds at each position of $\sigma_0$, follows from property 1).
    Now the second property is: The word corresponding to any model that satisfies the formula must belong to $\{u_ip\overline{v_i}q \mid 1 \leq i \leq k\}^\omega$. This can be easily expressed as a formula in LRC($\mathcal{D}$). 
    
    \item The value of $x$ at every $\Sigma$-position is strictly less than the value of $x_i$ at every position ahead of it, for all $1 \leq i \leq r$.
    The value of $x$ at every $\overline{\Sigma}$-position is strictly less than the value of $y_i$ at every position ahead of it, for all $1 \leq i \leq r$. This condition can be expressed as a formula in LRC($\mathcal{D}$). This ensures that the sequence of values of $x$ at $\Sigma$-positions is in ascending order and that the sequence of values of $x$ at $\overline{\Sigma}$-positions is also in ascending order.
    
    \item For all $1 \leq i \leq r$, the value of $x$ at every $a_i$ position is equal to the value of $x$ at some $\overline{a_i}$ position (in the past or future). Similarly, for all $1 \leq i \leq r$, the value of $x$ at every $a_i$ position is equal to the value of $x$ at some $\overline{a_i}$ position (in the past or future). 
    
\end{enumerate}

(P3) and (P4), together ensure that for any model satisfying the formula, the value of $x$ at the $i$th $\Sigma$ position is equal to the value of $x$ at the $i$th $\overline{\Sigma}$ position, for all $i \geq 0$ (Note that the values of $x$ at all $\Sigma$ positions are different and that the values of $x$ at all $\overline{\Sigma}$ positions are also different). (P4) says that the value of $x$ at an $a_j$ position is equal to the value of $x$ at an $\overline{a_j}$ position (for all $1 \leq j \leq r$). This means that if the letter at the $i$th $\Sigma$ position is $a_j$, then the letter at $i$th $\overline{\Sigma}$ position is $\overline{a_j}$ (for all $1 \leq j \leq r$, $i \geq 0$). \\

Thus it is clear that if $\sigma$ is a model satisfying the properties above, then the word corresponding to $\sigma$ is of the form $u_{i_0}p\overline{v_{i_0}}qu_{i_1}p\overline{v_{i_1}}q...$ where $u_{i_0}u_{i_1}... = v_{i_0}v_{i_1}...$. Hence we have a solution to the infinite PCP. Similarly, now if $w$ is a solution to the infinite PCP, consider the word $\hat{w}$ as defined earlier. Now consider a $D$-valuation sequence $\sigma$ that we defined using $\hat{w}$, satisfying conditions (C1), (C2) and (C3). In addition, we also wish to ensure that $\sigma$ satisfies property (P3), which can be done by choosing $\sigma$ accordingly. Now any such $\sigma$ defined using $\hat{w}$ that satisfies (C1), (C2), (C3) and (P3), clearly satisfies all the four properties ((P1) - (P4)) and hence is a model of the formula.\\

All that remains is to describe the formula $\psi$ which is a conjunction of the four properties described above.

\begin{equation*}
    \psi = \psi_1 \wedge \psi_2 \wedge \psi_3 \wedge \psi_4
\end{equation*}

where:
\begin{align*}
    \psi_1 = \textnormal{\textbf{G}}\begin{pmatrix}&\bigvee\limits_{i=1}^{r} \begin{pmatrix}&\left(\left(\left(x \approx x_i \right) \wedge \bigwedge\limits_{i \not = j} \neg\left(x \approx x_j \right) \wedge \bigwedge\limits_{j=1}^{r} \neg\left(x \approx y_j \right)\right) \right. \\
    & \left. \vee  \left(\left(x \approx y_i \right) \wedge \bigwedge\limits_{i \not = j} \neg\left(x \approx y_j \right) \wedge \bigwedge\limits_{j=1}^{r} \neg\left(x \approx x_j \right)\right)\right)
    \wedge \neg\left(x \approx w_1 \right) \wedge \neg\left(x \approx w_2 \right)\end{pmatrix} \\\\
    & \left. \vee \left(\left(x \approx w_1 \right) \wedge \bigwedge\limits_{j=1}^{r} \neg\left(x \approx x_j \right) \wedge \bigwedge\limits_{j=1}^{r} \neg\left(x \approx y_j \right) \wedge \neg\left(x \approx w_2 \right)\right) \right. \\\\
    & \left. \vee \left(\left(x \approx w_2 \right) \wedge \bigwedge\limits_{j=1}^{r} \neg\left(x \approx x_j \right) \wedge \bigwedge\limits_{j=1}^{r} \neg\left(x \approx y_j \right) \wedge \neg\left(x \approx w_1 \right)\right) \right. \end{pmatrix}  \\
\end{align*}

\begin{flalign*}
    \psi_2 = \alpha \wedge \textnormal{\textbf{G}}(q \Rightarrow \textnormal{\textbf{X}} \alpha) &&
\end{flalign*}
\newline
with 
\begin{flalign*}
    \alpha = \bigvee\limits_{i=1}^{k}\begin{pmatrix}& u_{i_1} \wedge \textnormal{\textbf{X}} u_{i_2} \wedge ... \wedge \textnormal{\textbf{X}}^{(l_i - 1)} u_{i_{l_i}} \\\\
    & \wedge \textnormal{\textbf{X}}^{l_i}p \wedge \textnormal{\textbf{X}}^{(l_i + 1)} \overline{v_{i_1}} \\\\
    & \wedge 
    \textnormal{\textbf{X}}^{(l_i + 2)} \overline{v_{i_2}} \wedge ... \wedge \textnormal{\textbf{X}}^{(l_i + m_i)} \overline{v_{i_{m_i}}} \\\\
    & \wedge 
    \textnormal{\textbf{X}}^{(l_i + m_i + 1)}q 
    \end{pmatrix} &&
\end{flalign*}
\newline
(Here, $l_i$ denotes the length of the word $u_i$ and $u_{i_j}$ denotes the $j$th letter of the word $u_i$ for all $1 \leq j \leq l_i$. Similarly, $m_i$ denotes the length of the word $\overline{v_i}$ and $\overline{v_{i_j}}$ denotes the $j$th letter of the word $\overline{v_i}$ for all $1 \leq j \leq m_i$.)

\begin{flalign*}
\psi_3 = \textnormal{\textbf{G}}\begin{pmatrix} &\bigwedge\limits_{i,j} \neg((x_i \approx \textnormal{\textbf{XF}} x_j) \wedge (x \approx x_i)) \wedge \bigwedge\limits_{i,j} \neg((x_i > \textnormal{\textbf{XF}} x_j) \wedge (x \approx x_i)) \\\\
& \wedge \bigwedge\limits_{i,j} \neg((y_i \approx \textnormal{\textbf{XF}} y_j) \wedge (x \approx y_i)) \wedge \bigwedge\limits_{i,j} \neg((y_i > \textnormal{\textbf{XF}} y_j) \wedge (x \approx y_i))
\end{pmatrix} &&
\end{flalign*}

\begin{flalign*}
\psi_4 = \textnormal{\textbf{G}}\begin{pmatrix}&\bigwedge\limits_{i=1}^{r}((x \approx x_i) \Rightarrow ((x \approx \langle (x \approx y_i)? \rangle x) \vee (x \approx {\langle (x \approx y_i)? \rangle}^{-1} x))) \\\\
& \wedge \bigwedge\limits_{i=1}^{r}((x \approx y_i) \Rightarrow ((x \approx \langle (x \approx x_i)? \rangle x) \vee (x \approx {\langle (x \approx x_i)? \rangle}^{-1} x)))
\end{pmatrix} && 
\end{flalign*}
\newline
One can easily verify that the formulas $\psi_1$-$\psi_4$ express the properties (P1)-(P4). \\

Hence, given an instance of infinite PCP, we have constructed an LRC$(\mathcal{D})$ formula $\psi$ such that the infinite PCP has a solution iff the formula $\psi$ is satisfiable. \\
\begin{theorem}
The satisfiability problems for LRC$(\mathcal{D})$ and LRC$^{\top}(\mathcal{D})$ are undecidable.
\end{theorem}

\section{CLRV$^{\textnormal{XF},\textnormal{XF}^{-1}}$ with interactions}
\subsection{Introduction}
Previously we had defined the logic CLRV$^{\textnormal{XF}}$($\mathcal{D}$) (where the constraint system $\mathcal{D}$ is of the form $(D, <, \approx)$, where $D$ is an infinite set, `$<$' is any strict total ordering over $D$ and `$\approx$' is the equality relation). The syntax of the logic was as follows:
\begin{align*}
     \phi :: = \hspace{0.06 cm} & x_{\text{loc}} \approx \textnormal{\textbf{X}}^{i} y_{\text{loc}} \hspace{0.06cm} |\hspace{0.06cm} x_{\text{loc}} < \textnormal{\textbf{X}}^{i} y_{\text{loc}} \hspace{0.06cm}|\hspace{0.06cm} x_{\textnormal{rem}} \approx \textnormal{\textbf{X}}^{i} y_{\textnormal{rem}\>} \hspace{0.06cm} | \hspace{0.06cm} x_{\text{rem}} \approx \textnormal{\textbf{XF}} y_{\text{rem}} \hspace{0.06cm}  \\ 
     &| \hspace{0.06cm} \phi \wedge \phi \hspace{0.06cm}
     | \hspace{0.06cm} \neg \phi \hspace{0.06cm} | \hspace{0.06cm} \textnormal{\textbf{X}} \phi \hspace{0.06cm} | \hspace{0.06cm} \phi\textnormal{\textbf{U}}\phi \hspace{0.06cm} | \hspace{0.06cm} \textnormal{\textbf{X}}^{-1} \phi 
     |\hspace{0.06cm} \phi\textnormal{\textbf{S}}\phi, \\ & \text{ where } x_{\text{loc}}, y_{\text{loc}} \in V_{\text{local}},\hspace{0.12cm} x_{\text{rem}}, y_{\text{rem}} \in V_{\text{remote}}
\end{align*}
Here the set of variables $V = V_{\text{local}} \biguplus V_{\text{remote}}$ and the local constraints involve
only variables in $V_{\text{local}}$ and remote constraints involve only variables in $V_{\text{remote}}$.\\

Now we consider an extension of CLRV$^{\textnormal{XF}}$($\mathcal{D}$) that includes past remote constraints. We call this extension CLRV$^{\textnormal{XF},\textnormal{XF}^{-1}}$($\mathcal{D}$). Formally, the syntax of this logic is as follows:
\begin{align*}
     \phi :: = \hspace{0.06 cm} & x_{\text{loc}} \approx \textnormal{\textbf{X}}^{i} y_{\text{loc}} \hspace{0.06cm} |\hspace{0.06cm} x_{\text{loc}} < \textnormal{\textbf{X}}^{i} y_{\text{loc}} \hspace{0.06cm}|\hspace{0.06cm} x_{\textnormal{rem}} \approx \textnormal{\textbf{X}}^{i} y_{\textnormal{rem}\>} \hspace{0.06cm} | \hspace{0.06cm} x_{\text{rem}} \approx \textnormal{\textbf{XF}} y_{\text{rem}} \hspace{0.06cm}  \\ 
     &| \hspace{0.06cm} x_{\text{rem}} \approx \textnormal{\textbf{XF}}^{-1} y_{\text{rem}} \hspace{0.06cm} |\hspace{0.06cm} \phi \wedge \phi \hspace{0.06cm}
     | \hspace{0.06cm} \neg \phi \hspace{0.06cm} | \hspace{0.06cm} \textnormal{\textbf{X}} \phi \hspace{0.06cm} | \hspace{0.06cm} \phi\textnormal{\textbf{U}}\phi \hspace{0.06cm} | \hspace{0.06cm} \textnormal{\textbf{X}}^{-1} \phi 
     |\hspace{0.06cm} \phi\textnormal{\textbf{S}}\phi, \\ & \text{ where } x_{\text{loc}}, y_{\text{loc}} \in V_{\text{local}},\hspace{0.12cm} x_{\text{rem}}, y_{\text{rem}} \in V_{\text{remote}}
\end{align*}
We have already shown that the satisfiability problem for CLRV$^{\textnormal{XF}}$($\mathcal{D}$) is decidable for dense and open domains. The proof of decidability depended on the proofs of decidability of CLTL($\mathcal{D}$) and CLTL$^{\textnormal{XF}}$. Now, it has also been proven in \cite{DDG12} that the logic CLTL$^{\textnormal{XF},\textnormal{XF}^{-1}}$ is decidable. So, in a way, similar to how we proved that the satsifiability problem for CLRV$^{\textnormal{XF}}$($\mathcal{D}$) is decidable for dense and open constraint systems, one can combine the proofs of decidability of CLTL($\mathcal{D}$) and CLTL$^{\textnormal{XF},\textnormal{XF}^{-1}}$ to prove that the satisfiability problem for CLRV$^{\textnormal{XF},\textnormal{XF}^{-1}}$($\mathcal{D}$) is decidable for dense and open domains.\\

Now we consider a further extension of CLRV$^{\textnormal{XF},\textnormal{XF}^{-1}}$($\mathcal{D}$). In this extension, we allow for some interaction between local and remote variables. Intuitively, the simplest possible interaction would be to allow equality check between local and remote variables. We consider this extension first and call it CLRV$^{\textnormal{XF},\textnormal{XF}^{-1}, \textnormal{int}}$($\mathcal{D}$). Formally, the syntax of this extension is as follows:
\begin{align*}
     \phi :: = \hspace{0.06 cm} & x_{\text{loc}} \approx \textnormal{\textbf{X}}^{i} y_{\text{loc}} \hspace{0.06cm} |\hspace{0.06cm} x_{\text{loc}} < \textnormal{\textbf{X}}^{i} y_{\text{loc}} \hspace{0.06cm}|\hspace{0.06cm} x_{\textnormal{rem}} \approx \textnormal{\textbf{X}}^{i} y_{\textnormal{rem}\>} \hspace{0.06cm} | \hspace{0.06cm} x_{\text{rem}} \approx \textnormal{\textbf{XF}} y_{\text{rem}} \hspace{0.06cm}  \\ 
     &| \hspace{0.06cm} x_{\text{rem}} \approx \textnormal{\textbf{XF}}^{-1} y_{\text{rem}} \hspace{0.06cm} |\hspace{0.06cm} x_{\text{loc}} \approx  y_{\text{rem}} \hspace{0.06cm} | \hspace{0.06cm} \phi \wedge \phi \hspace{0.06cm}
     | \hspace{0.06cm} \neg \phi \hspace{0.06cm} | \hspace{0.06cm} \textnormal{\textbf{X}} \phi \hspace{0.06cm} | \hspace{0.06cm} \phi\textnormal{\textbf{U}}\phi \hspace{0.06cm} | \hspace{0.06cm} \textnormal{\textbf{X}}^{-1} \phi 
     |\hspace{0.06cm} \phi\textnormal{\textbf{S}}\phi, \\ & \text{ where } x_{\text{loc}}, y_{\text{loc}} \in V_{\text{local}},\hspace{0.12cm} x_{\text{rem}}, y_{\text{rem}} \in V_{\text{remote}}
\end{align*}

Here the constraint system $\mathcal{D}$ is of the form $(D, <, \approx)$, where $D$ is an infinite set, `$<$' is any strict total ordering over $D$ and `$\approx$' is the equality relation.\\

A $D$-valuation is defined to be a map from $V$ to $D$. A model is an infinite sequence $\sigma$ of $D$-valuations.\\

The satisfaction relation extends the satisfaction relation as defined for CLRV$^{\textnormal{XF}}$($\mathcal{D}$). For the atomic formulae not present in CLRV$^{\textnormal{XF}}$($\mathcal{D}$), the satisfaction relation $\models$ is defined, for every model $\sigma$ and for all $i \geq 0$, as follows:
\begin{itemize}
    \item $\sigma, i \models x_\text{rem} \approx \textnormal{\textbf{XF}}^{-1} y_\text{rem}$ iff there exists $j \geq 0$ such that $j < i$ and $\sigma(i)(x_\text{rem}) = \sigma(j)(y_\text{rem})$
    
    \item $\sigma, i \models x_\text{loc} \approx  y_\text{rem}$ iff there exists $\sigma(i)(x_\text{loc}) = \sigma(j)(y_\text{rem})$
\end{itemize}

One could extend CLRV$^{\textnormal{XF},\textnormal{XF}^{-1}, \textnormal{int}}$($\mathcal{D}$) by replacing $x_{\text{rem}} \approx \textnormal{\textbf{XF}} y_{\text{rem}}$ with $x_{\text{rem}} \approx \langle \phi? \rangle y_{\text{rem}}$ and $x_{\text{rem}} \approx \textnormal{\textbf{XF}}^{-1} y_{\text{rem}}$ with $x_{\text{rem}} \approx {\langle \phi? \rangle}^{-1} y_{\text{rem}}$ in the syntax. Formally, the syntax of this logic is: 
\begin{align*}
     \phi :: = \hspace{0.06 cm} & x_{\text{loc}} \approx \textnormal{\textbf{X}}^{i} y_{\text{loc}} \hspace{0.06cm} |\hspace{0.06cm} x_{\text{loc}} < \textnormal{\textbf{X}}^{i} y_{\text{loc}} \hspace{0.06cm}|\hspace{0.06cm} x_{\textnormal{rem}} \approx \textnormal{\textbf{X}}^{i} y_{\textnormal{rem}\>} \hspace{0.06cm} |\> x_{\textnormal{rem}} \approx \langle \phi? \rangle y_{\textnormal{rem}} \\
     &\>|\> x_{\textnormal{rem}} \approx {\langle \phi? \rangle}^{-1} y_{\textnormal{rem}} \>|\hspace{0.06cm} x_{\text{loc}} \approx  y_{\text{rem}} \hspace{0.06cm} | \hspace{0.06cm} \phi \wedge \phi \hspace{0.06cm}
     | \hspace{0.06cm} \neg \phi \hspace{0.06cm} | \hspace{0.06cm} \textnormal{\textbf{X}} \phi \hspace{0.06cm} | \hspace{0.06cm} \phi\textnormal{\textbf{U}}\phi \hspace{0.06cm} | \hspace{0.06cm} \textnormal{\textbf{X}}^{-1} \phi 
     |\hspace{0.06cm} \phi\textnormal{\textbf{S}}\phi, \\ & \text{ where } x_{\text{loc}}, y_{\text{loc}} \in V_{\text{local}},\hspace{0.12cm} x_{\text{rem}}, y_{\text{rem}} \in V_{\text{remote}}
\end{align*}

The satisfaction relation for the new atomic formulae is defined as follows:
\begin{itemize}
    \item $\sigma, i \models x_{\textnormal{rem}} \approx \langle \phi? \rangle y_{\textnormal{rem}}$ iff there exists $j$ such that $i < j$, $\sigma(i)(x_{\textnormal{rem}}) = \sigma(j)(y_{\textnormal{rem}})$ and $\sigma, j \models \phi$
    
    \item $\sigma, i \models x_{\textnormal{rem}} \approx {\langle \phi? \rangle}^{-1} y_{\textnormal{rem}}$ iff there exists $j \geq 0$ such that $j < i$, $\sigma(i)(x_{\textnormal{rem}}) = \sigma(j)(y_{\textnormal{rem}})$ and $\sigma, j \models \phi$
\end{itemize}

It turns out that there is a polynomial-time algorithm, which, given a formula $\phi$ in this extension, computes a formula $\phi^{\prime}$ in  CLRV$^{\textnormal{XF},\textnormal{XF}^{-1}, \textnormal{int}}$($\mathcal{D}$) that preserves satisfiability: there is a model
$\sigma$ such that $\sigma \models \phi$ iff there is a model $\sigma^{\prime}$ such that $\sigma^{\prime} \models \phi^{\prime}$. As in the case of LRC($\mathcal{D}$), the proof follows from the Proof of Corollary 4.6 of \cite{SDM16}. 

\begin{problem}
Let $\mathcal{D} = (D, <, \approx)$ be a constraint system. The satisfiability problem for CLRV$^{\textnormal{XF},\textnormal{XF}^{-1}, \textnormal{int}}$($\mathcal{D}$) is: Given an CLRV$^{\textnormal{XF},\textnormal{XF}^{-1}, \textnormal{int}}$($\mathcal{D}$) formula $\phi$, does there exist a model which satisfies $\phi$?
\end{problem}

We shall now prove that the satisfiability problem for CLRV$^{\textnormal{XF},\textnormal{XF}^{-1}, \textnormal{int}}$($\mathcal{D}$) is undecidable. 

\subsection{Proof of Undecidability}
Here, we again give a reduction from infinite version of PCP. Since the problem is that of satisfiability, we use the formulae $x_{\text{rem}} \approx \langle \phi? \rangle y_{\text{rem}}$ and $x_{\text{rem}} \approx {\langle \phi? \rangle}^{-1} y_{\text{rem}}$ freely while giving the reduction as we already saw that these can be expressed in terms of the formulae in CLRV$^{\textnormal{XF},\textnormal{XF}^{-1}, \textnormal{int}}$($\mathcal{D}$) while preserving satisfiability.\\

Consider a finite alphabet $\Sigma = \{a_1,a_2,...a_r\}$. An instance of infinite PCP
consists of $k$ pairs $(u_i, v_i)$ of words from $\Sigma^{*}$ and the question is whether there exists an infinite sequence of indices $i_0, i_1, ...$ such that $u_{i_0}u_{i_1}... = v_{i_0}v_{i_1}... $.\\

Let $\overline{\Sigma} = \Sigma \cup \Sigma^{\prime}$  be the alphabet consisting of two disjoint copies of $\Sigma$. If $w$ is a word in $\Sigma^{\omega}$, then $\overline{w} \in \overline{\Sigma}^{\omega}$ is obtained from $w$ by replacing each letter $a$ with the corresponding letter $\overline{a}$ in $\overline{\Sigma}$. Consider a solution $i_0, i_1, ... $ for the given infinite PCP instance. Let $w$ be the infinite word $u_{i_0}u_{i_1}, ... $, or equivalently the word $v_{i_0}v_{i_1}, ... $. Let $\hat{w}$ be the infinite word $u_{i_0}\overline{v_{i_0}}u_{i_1}\overline{v_{i_1}}...$ over $\overline{\Sigma}$.\\

Now, consider the set of remote variables $V_{\text{remote}} = \{x, x_1, x_2, ... x_r, y_1, y_2, ... y_r, w_1, w_2\}$ and the set of local variables $V_{\text{local}} = \{z_1, z_2. z_3, z_4\}$ (where $r = |\Sigma|$). As we already know, the set of all variables $V = V_{\text{local}} \biguplus V_{\text{remote}}$. Given a solution $\hat{w}$ to an instance of the infinite PCP, we encode it as a $D$-valuation sequence $\sigma$ over $V$, satisfying the following conditions:
\begin{enumerate}[(C1)]

    \item For all $i \geq 0$, $1 \leq j \leq r$:  
    
    $\hat{w}(i) = a_j$ iff $\sigma(i)(x) = \sigma(i)(x_j)$ and $\sigma(i)(x) \not = \sigma(i)(v)$ for any $v \in V_{\text{remote}}\setminus\{x_j\}$;\\
    
    $\hat{w}(i) = \overline{a_j}$ iff $\sigma(i)(x) = \sigma(i)(y_j)$ and $\sigma(i)(x) \not = \sigma(i)(v)$ for any $v \in V_{\text{remote}}\setminus\{y_j\}$
    
    \item The subsequence of values of the variable $x$ at $\Sigma$-positions (i.e. positions $i$ for which $\hat{w}(i) \in \Sigma$) is equal to the subsequence of values of the variable $x$ at $\overline{\Sigma}$-positions (i.e. positions $i$ for which $\hat{w}(i) \in \overline{\Sigma}$) 
    
    \item For all $i \geq 0$, $j \geq 0$, $i \not = j$, if $\hat{w}(i) \in \Sigma$ and $\hat{w}(j) \in \Sigma$, then $\sigma(i)(x) \not = \sigma(j)(x)$. Similarly, for all $i \geq 0$, $j \geq 0$, $i \not = j$, if $\hat{w}(i) \in \overline{\Sigma}$ and $\hat{w}(j) \in \overline{\Sigma}$, then $\sigma(i)(x) \not = \sigma(j)(x)$
    
\end{enumerate}
We now describe an CLRV$^{\textnormal{XF},\textnormal{XF}^{-1}, \textnormal{int}}$($\mathcal{D}$) formula $\psi$ such that $w$ is a solution to the infinite PCP if and only if $\sigma$ is a model of $\psi$. Note that while describing the formula we shall interchangeably use $(x \approx x_i)$ with $a_i$, $(x \approx y_i)$ with $\overline{a_i}$, $(x \approx w_1)$ with $p$ and $(x \approx w_2)$ with $q$. The formula is the conjunction of the following properties:

\begin{enumerate} [(P1)]
    \item At every position, the value of $x$ is equal to that of exactly one variable in $V_{\text{remote}} \setminus \{x\}$
    
    \item Now that we are taking a conjunction, we can define the notion of a word corresponding to any $D$-valuation sequence satisfying property 1. If $\sigma_0$ is a $D$-valuation sequence satisfying property 1, then the word corresponding to $\sigma_0$ is the infinite word $w_0 \in (\overline{\Sigma} \cup \{p,q\})^\omega$ such that at every position $i$, we have $w_0(i)$ equal to the unique letter from $\overline{\Sigma} \cup \{p,q\}$, that holds at $\sigma_0(i)$. (The existence and uniqueness of the letter that holds at each position of $\sigma_0$, follows from property 1).
    Now the second property is: The word corresponding to any model that satisfies the formula must belong to $\{u_ip\overline{v_i}q \mid 1 \leq i \leq k\}^\omega$. This can be easily expressed as a formula in CLRV$^{\textnormal{XF},\textnormal{XF}^{-1}, \textnormal{int}}$($\mathcal{D}$). 
    
    \item The value of $x$ at every $\Sigma$-position is strictly less than the value of $x$ at the next $\Sigma$-position.
    The value of $x$ at every $\overline{\Sigma}$-position is strictly less than the value of $x$ at the next $\overline{\Sigma}$-position. This condition can be expressed as a formula in CLRV$^{\textnormal{XF},\textnormal{XF}^{-1}, \textnormal{int}}$($\mathcal{D}$). This ensures that the sequence of values of $x$ at $\Sigma$-positions is in ascending order and that the sequence of values of $x$ at $\overline{\Sigma}$-positions is also in ascending order.
    
    \item For all $1 \leq i \leq r$, the value of $x$ at every $a_i$ position is equal to the value of $x$ at some $\overline{a_i}$ position (in the past or future). Similarly, for all $1 \leq i \leq r$, the value of $x$ at every $a_i$ position is equal to the value of $x$ at some $\overline{a_i}$ position (in the past or future). 
    
\end{enumerate}
These properties are very similar to the properties in the previous undecidability proof. Here, (P3) and (P4), together ensure that for any model satisfying the formula, the value of $x$ at the $i$th $\Sigma$ position is equal to the value of $x$ at the $i$th $\overline{\Sigma}$ position, for all $i \geq 0$ (Note that the values of $x$ at all $\Sigma$ positions are different and that the values of $x$ at all $\overline{\Sigma}$ positions are also different). (P4) says that the value of $x$ at an $a_j$ position is equal to the value of $x$ at an $\overline{a_j}$ position (for all $1 \leq j \leq r$). This means that if the letter at the $i$th $\Sigma$ position is $a_j$, then the letter at $i$th $\overline{\Sigma}$ position is $\overline{a_j}$ (for all $1 \leq j \leq r$, $i \geq 0$). \\

Thus it is clear that if $\sigma$ is a model satisfying the properties above, then the word corresponding to $\sigma$ is of the form $u_{i_0}p\overline{v_{i_0}}qu_{i_1}p\overline{v_{i_1}}q...$ where $u_{i_0}u_{i_1}... = v_{i_0}v_{i_1}...$. Hence we have a solution to the infinite PCP. Similarly, now if $w$ is a solution to the infinite PCP, consider the word $\hat{w}$ as defined earlier. Now consider a $D$-valuation sequence $\sigma$ that we defined using $\hat{w}$, satisfying conditions (C1), (C2) and (C3). Now any such $\sigma$ defined using $\hat{w}$ that satisfies (C1), (C2) and (C3), clearly satisfies all the four properties ((P1) - (P4)) and hence is a model of the formula.\\

All that remains is to describe the formula $\psi$ which is a conjunction of the four properties described above.

\begin{equation*}
    \psi = \psi_1 \wedge \psi_2 \wedge \psi_3 \wedge \psi_4
\end{equation*}

where:
\begin{align*}
    \psi_1 = \textnormal{\textbf{G}}\begin{pmatrix}&\bigvee\limits_{i=1}^{r} \begin{pmatrix}&\left(\left(\left(x \approx x_i \right) \wedge \bigwedge\limits_{i \not = j} \neg\left(x \approx x_j \right) \wedge \bigwedge\limits_{j=1}^{r} \neg\left(x \approx y_j \right)\right) \right. \\
    & \left. \vee  \left(\left(x \approx y_i \right) \wedge \bigwedge\limits_{i \not = j} \neg\left(x \approx y_j \right) \wedge \bigwedge\limits_{j=1}^{r} \neg\left(x \approx x_j \right)\right)\right)
    \wedge \neg\left(x \approx w_1 \right) \wedge \neg\left(x \approx w_2 \right)\end{pmatrix} \\\\
    & \left. \vee \left(\left(x \approx w_1 \right) \wedge \bigwedge\limits_{j=1}^{r} \neg\left(x \approx x_j \right) \wedge \bigwedge\limits_{j=1}^{r} \neg\left(x \approx y_j \right) \wedge \neg\left(x \approx w_2 \right)\right) \right. \\\\
    & \left. \vee \left(\left(x \approx w_2 \right) \wedge \bigwedge\limits_{j=1}^{r} \neg\left(x \approx x_j \right) \wedge \bigwedge\limits_{j=1}^{r} \neg\left(x \approx y_j \right) \wedge \neg\left(x \approx w_1 \right)\right) \right. \end{pmatrix}  \\
\end{align*}

\begin{flalign*}
    \psi_2 = \alpha \wedge \textnormal{\textbf{G}}(q \Rightarrow \textnormal{\textbf{X}} \alpha) &&
\end{flalign*}
\newline
with 
\begin{flalign*}
    \alpha = \bigvee\limits_{i=1}^{k}\begin{pmatrix}& u_{i_1} \wedge \textnormal{\textbf{X}} u_{i_2} \wedge ... \wedge \textnormal{\textbf{X}}^{(l_i - 1)} u_{i_{l_i}} \\\\
    & \wedge \textnormal{\textbf{X}}^{l_i}p \wedge \textnormal{\textbf{X}}^{(l_i + 1)} \overline{v_{i_1}} \\\\
    & \wedge 
    \textnormal{\textbf{X}}^{(l_i + 2)} \overline{v_{i_2}} \wedge ... \wedge \textnormal{\textbf{X}}^{(l_i + m_i)} \overline{v_{i_{m_i}}} \\\\
    & \wedge 
    \textnormal{\textbf{X}}^{(l_i + m_i + 1)}q 
    \end{pmatrix} &&
\end{flalign*}
\newline
(Here, $l_i$ denotes the length of the word $u_i$ and $u_{i_j}$ denotes the $j$th letter of the word $u_i$ for all $1 \leq j \leq l_i$. Similarly, $m_i$ denotes the length of the word $\overline{v_i}$ and $\overline{v_{i_j}}$ denotes the $j$th letter of the word $\overline{v_i}$ for all $1 \leq j \leq m_i$.)

\begin{flalign*}
    \psi_3 = \beta \wedge \gamma &&
\end{flalign*}
with 
\begin{flalign*}
\beta = \textnormal{\textbf{G}}\begin{pmatrix} &\bigwedge\limits_{i=1}^{r} \begin{pmatrix}((x \approx y_i) \wedge \textnormal{\textbf{X}} (\neg q)) \Rightarrow ((x \approx z_1) \wedge (z_1 < \textnormal{\textbf{X}}z_2) \wedge (\textnormal{\textbf{X}}x \approx \textnormal{\textbf{X}}z_2)) \end{pmatrix} \\\\
 &\wedge \bigwedge\limits_{i=1}^{r} \begin{pmatrix}((x \approx y_i) \wedge \textnormal{\textbf{X}} q) \Rightarrow \begin{pmatrix}&(x \approx z_1) \wedge (z_1 < \textnormal{\textbf{X}}z_2) \\\\
 &\wedge \textnormal{\textbf{X}}\begin{pmatrix}&((z_2 \approx \textnormal{\textbf{X}}z_2) \wedge \neg p) \\
 & \textnormal{\textbf{U}} (p \wedge (z_2 \approx \textnormal{\textbf{X}}z_2) \wedge \textnormal{\textbf{X}}(z_2 \approx x)))\end{pmatrix}\end{pmatrix} \end{pmatrix} \\\\
\end{pmatrix} &&
\end{flalign*}
and 
\begin{flalign*}
\gamma = \textnormal{\textbf{G}}\begin{pmatrix} &\bigwedge\limits_{i=1}^{r} \begin{pmatrix}((x \approx x_i) \wedge \textnormal{\textbf{X}} (\neg p)) \Rightarrow ((x \approx z_3) \wedge (z_3 < \textnormal{\textbf{X}}z_4) \wedge (\textnormal{\textbf{X}}x \approx \textnormal{\textbf{X}}z_4)) \end{pmatrix} \\\\
 &\wedge \bigwedge\limits_{i=1}^{r} \begin{pmatrix}((x \approx x_i) \wedge \textnormal{\textbf{X}} p) \Rightarrow \begin{pmatrix}&(x \approx z_3) \wedge (z_3 < \textnormal{\textbf{X}}z_4) \\\\
 &\wedge \textnormal{\textbf{X}}\begin{pmatrix}&((z_4 \approx \textnormal{\textbf{X}}z_4) \wedge \neg q) \\
 & \textnormal{\textbf{U}} (q \wedge (z_4 \approx \textnormal{\textbf{X}}z_4) \wedge \textnormal{\textbf{X}}(z_4 \approx x)))\end{pmatrix}\end{pmatrix} \end{pmatrix} \\\\
\end{pmatrix} &&
\end{flalign*}
Note here that we are able to say that the value of $x$ at the current $\Sigma$-position (which is a remote variable) is less than the value of $x$ (a remote variable) at the next $\Sigma$-position. That is, we are able to compare the values of remote variables at positions which are not at a fixed distance. This is possible only because we allow interaction between the remote variable $x$ and the local variables $z_1$ and $z_2$. 
\begin{flalign*}
\psi_4 = \textnormal{\textbf{G}}\begin{pmatrix}&\bigwedge\limits_{i=1}^{r}((x \approx x_i) \Rightarrow ((x \approx \langle (x \approx y_i)? \rangle x) \vee (x \approx {\langle (x \approx y_i)? \rangle}^{-1} x))) \\\\
& \wedge \bigwedge\limits_{i=1}^{r}((x \approx y_i) \Rightarrow ((x \approx \langle (x \approx x_i)? \rangle x) \vee (x \approx {\langle (x \approx x_i)? \rangle}^{-1} x)))
\end{pmatrix} && 
\end{flalign*}
One can easily verify that the formulas $\psi_1$-$\psi_4$ express the properties (P1)-(P4). \\

Hence, given an instance of infinite PCP, we have constructed an CLRV$^{\textnormal{XF},\textnormal{XF}^{-1}, \textnormal{int}}$($\mathcal{D}$)  formula $\psi$ such that the infinite PCP has a solution iff the formula $\psi$ is satisfiable. \\

\begin{theorem}
The satisfiability problem for CLRV$^{\textnormal{XF},\textnormal{XF}^{-1}, \textnormal{int}}$($\mathcal{D}$) is undecidable.\\
\end{theorem}

It is worth noting here, why the proof of decidability of CLRV$^{\textnormal{XF}}$($\mathcal{D}$) over dense and open domains, does not go through in this case. Remember that one of the main ideas in the proof of decidability was to split the symbolic model into its local and remote projections and construct an automaton which accepts those symbolic models whose local projections and remote projections are independently realizable. It was based on Lemma 2.1.1. But such a lemma does not hold here as there is an interaction between the local and remote variables. Hence no obvious extrapolation of the automata-based construction in the proof of decidability of CLRV$^{\textnormal{XF}}$($\mathcal{D}$) would work in this case.
    \chapter*{Conclusion} 
\addcontentsline{toc}{chapter}{Conclusion}

In this thesis, we have considered the logic CLRV$^{\textnormal{XF}}$($\mathcal{D}$) which is a natural extension of the logics CLTL($\mathcal{D}$) (\cite{DD07}) and CLTL$^{\textnormal{XF}}$ (\cite{DDG12}). It is known that the satisfiabiity problem for the logic CLTL($\mathcal{D}$) is decidable over dense, open domains (Theorem 1.1.3) and the satisfiability problem for CLTL$^{\textnormal{XF}}$ is also decidable (Theorem 1.2.1).  Combining the proofs of these decidability results, we proved in our thesis that over dense, open domains, the satisfiability problem for CLRV$^{\textnormal{XF}}$($\mathcal{D}$) is decidable.\\

Now, it is also known that the satisfiability problem for the logic CLTL($\mathcal{D}$) is decidable even over $\mathbb{N}$ and over $\mathbb{Z}$ (\cite[Theorem 6.4]{DD07}, \cite[Theorem 7.3]{DD07}). So a natural question to ask would be whether the satisfiability problem for CLRV$^{\textnormal{XF}}$($\mathcal{D}$) is also decidable over $\mathbb{N}$ and over $\mathbb{Z}$. This is still open and is one possible direction for future work.\\

We have proved the undecidability of the satisfiability problem for LRC$^{\top}(\mathcal{D})$ (Theorem 3.1.1). Note that the past remote constraint $x \approx \textnormal{\textbf{XF}}^{-1} y$ appears in the syntax of this logic. Also, one can observe that the formula $\psi_4$ in section 3.1 of our thesis uses this past remote constraint to express property (P4). An interesting question to ask at this point would be: Is the satisfiability problem for this logic undecidable even without this past remote constraint? This problem is also open.


\begin{thebibliography}{6}
\bibitem{DD07} 
St{\'e}phane Demri and Deepak D'Souza.
\textit{An automata-theoretic approach to constraint LTL}. 
Information and Computation 205 (2007) 380–415.

\bibitem{DDG12} 
St{\'e}phane Demri, Deepak D'Souza and R{\'e}gis Gascon.
\textit{Temporal Logic of Repeating Values}.
LFCS 2007: Logical Foundations of Computer Science pp 180-194.
 
\bibitem{VW86} 
 M. Vardi and P. Wolper.
\textit{ An automata theoretic approach to automatic program verification}. 
Logic in Computer
Science, IEEE, 1986, pp. 332–334

\bibitem{MCATL10}
\textit{Two-variable Logic on Data Words}.
ACM Transactions on Computational Logic, Volume 12, No. 4

\bibitem{SDM16}
\textit{REASONING ABOUT DATA REPETITIONS WITH COUNTER SYSTEMS}.
Logical Methods in Computer Science, Volume 12, Issue 3 (August 4, 2016) lmcs:1645

\bibitem{FP18}
\textit{Playing with Repetitions in Data Words Using Energy Games}
LICS '18: Proceedings of the 33rd Annual ACM/IEEE Symposium on Logic in Computer Science
\end{thebibliography}
\end{document}